\documentclass{article}

\usepackage{amsmath,amsfonts,amssymb,amsthm}
\usepackage{graphicx}
\usepackage[a4paper,margin=2.5cm]{geometry}

\newtheorem{theorem}{Theorem}[section]

\newtheorem{proposition}[theorem]{Proposition}

\newtheorem*{assumption}{Assumption}
\theoremstyle{definition}
\newtheorem{remark}[theorem]{Remark}

\numberwithin{equation}{section}

\DeclareMathOperator*{\slim}{s-lim}
\DeclareMathOperator{\Tr}{Tr}
\DeclareMathOperator{\Ran}{Ran}

\begin{document}

\begin{center}

{\Large Memory effects in non-interacting mesoscopic transport}

\vspace{0.5cm}

Horia D. Cornean\footnote{\small Department of Mathematical Sciences,
     Aalborg University, Fredrik Bajers Vej 7G, DK-9220 \O{} Aalborg, Denmark. 
     E-mail: \texttt{cornean@math.aau.dk}}, Arne Jensen\footnote{\small Department of Mathematical Sciences,
     Aalborg University, Fredrik Bajers Vej 7G,
     DK-9220 \O{} Aalborg, Denmark. 
     E-mail: \texttt{matarne@math.aau.dk}}
     and Gheorghe Nenciu\footnote{\small Institute of Mathematics  of the Romanian Academy,
 Research Unit 1, P.O. Box 1-764,
 RO-014700 Bucharest, Romania. 
 E-mail: \texttt{Gheorghe.Nenciu@imar.ro}}

\end{center}

\begin{abstract}
Consider a quantum dot coupled to two semi-infinite one-dimensional
leads at thermal equilibrium. We turn on adiabatically a bias between the leads such that  
there exists exactly one discrete eigenvalue both at the beginning and
at the end of the switching procedure. It is shown that 
the expectation on the final bound state strongly depends on the history of the switching procedure. 
On the contrary, the contribution to the final steady-state
corresponding to the continuous spectrum has no memory, and only depends on the initial and final 
values of the bias. 
\end{abstract}

\section{Introduction and the main results}\label{intro}
Memory effects in quantum transport are rather common in systems where the carriers have self-interactions, and they can be of several types. Both in the 
Master Equation approach \cite{KF1,KF2} and the TDFT approach \cite{PSC} one observes a dependence of the steady state current both on the initial state and the switching procedure. Very 
recently \cite{CM} such a dependence on the initial state of the sample was proved in the cotunneling regime. 

In contrast, in the non-interacting mesoscopic quantum transport one can show that the charge/energy {\it current} observables depend neither on the initial state of the sample, nor on the switching procedure.

The sudden switch of a coupling at $t=0$ (in an initially partitioned system) has been thoroughly investigated in a number of previous works, see for 
example \cite{AJPP,JEP,Jaksic4,Jaksic5,Nenciu}; for a more physical approach see \cite{Caroli}. The steady state currents 
(computed as Ces{\` a}ro means) are very robust,  irrespective of which method one uses in order to induce a non-equilibrium state in the system (see \cite{CNZ}): their values only 
depend on the initial equilibrium state 
and on the final expression of the Hamiltonian which governs the evolution after the switching procedure is over.

In the partition-free approach introduced by Cini \cite{Cini}, 
the situation is rather similar. In \cite{CDP} it is treated a situation in which a thermal equilibrium state  
is perturbed by turning on adiabatically a bias between the leads. 
It is shown that if the instantaneous discrete spectrum of the one-body Hamiltonian is always well separated from the instantaneous continuous spectrum, 
then the adiabatic limit of the current coincides with the steady state current value of the sudden switch (\cite{Stefa1,Stefa2,CGZ}). 

In the current paper we investigate the situation in which the instantaneous discrete spectrum can enter the continuous one. It turns out that 
at the adiabatic limit, the steady state value of the charge current still has no memory of the switching procedure. 
But the situation is totally different for the  expectation on the final bound states. The adiabatic limit of this expectation is 
highly dependent on whether the instantaneous discrete spectrum enters the continuous spectrum or not. From a mathematical point of view, 
this phenomenon is related with the so-called 'adiabatic pair creation', see \cite{Nenciu2,Nenciu3, PD}. 

  Our model is of Wigner-Weisskopf type, 
see \cite{JEP} for a rather complete spectral analysis. We choose this model because we want to maximize clarity and minimize the technicality of the proofs, 
but most of the results below can be generalized to samples containing more sites, or even to a continuous setting.

\subsection{The setting and notation}

We consider two semi-infinite discrete leads coupled to a small system consisting of just one site. 

The single-particle Hilbert space is 
${\mathcal H}=\mathbf{C}\oplus \{l^2(\mathbf{N}_-)\oplus l^2(\mathbf{N}_+)\}=:{\mathcal H}_S\oplus{\mathcal H}_L$. The canonical  basis in 
${\mathcal H}_L$ is denoted by $\{|i_{\gamma}\rangle:\; \gamma=\pm, \;i\geq 0\}$ where 
$i_{\gamma}$ is the $i$-th site of the lead $\gamma$. Similarly, we denote by $\{|S\rangle \}$ 
the basis element of $\mathbf{C}$. With these notations we introduce the single-particle Hamiltonians $h_\pm$ which describe an 
electron on the leads to be just two copies of 
the usual one-dimensional discrete Laplacean $L$ initially defined on $l^2(\mathbf{Z})$ and then 
restricted to $l^2(\mathbf{N})$ with Dirichlet boundary condition at $-1$. 
With a physicist's notation we have 
$$L=\sum_{j \in \mathbf{Z}} \{
|j+1\rangle\langle j| +|j-1\rangle\langle j| \}.$$
If $\Pi_\pm=\sum_{j \geq 0}|j_\pm\rangle\langle j_\pm|$ are the projections on the left/right leads, then by definition 
$h_\pm:=\Pi_\pm L\Pi_\pm$.  We introduce the operators:
\begin{align}
\label{hL}
h_L&:=\sum_{\gamma =\pm} h_{\gamma},\quad h_S:=E_0|S\rangle\langle S|, \quad 
h_T=\tau\sum_{\gamma=\pm} \{|0_{\gamma}\rangle\langle S|+|S\rangle\langle 0_{\gamma}|\},
\end{align}
where $E_0\geq 0$ and $0<|\tau|\leq 1$ are real parameters to be chosen later.  If $v\geq 0$ is another real parameter 
which models the potential bias between the two leads, then the total Hamiltonian reads as: 
\begin{equation}\label{H_p}
h(v):=h_S+h_L+v\Pi_- +h_T.
\end{equation} 
We write
\begin{equation}
h_0(v)=h_S+h_L+v\Pi_- ,
\end{equation}
such that $h(v)=h_0(v)+h_T$.

The spectrum of $h_\pm$ is absolutely continuous and equals $[-2,2]$. The continuous spectrum of $h_0(v)$ is
$\sigma_{\rm ac}(h_0(v))=[-2,2]\cup [-2+v,2+v]$ while the pure point part is independent of $v$ and given by 
$\sigma_{\rm pp}(h_0(v))=\{E_0\}$. 

In Section \ref{appendix} we will treat in great detail the spectral
properties of $h(v)$ as a function of $v$. 
In particular, we will show 
that if $E_0$ is large enough and $\tau$ small enough, then $h(v)$ has
a unique discrete eigenvalue 
$\lambda(v)$ in the 
interval $(v+2,\infty)$ as long as $v$ is strictly smaller than a
critical value $v_{c,1}$. When $v= v_{c,1}$, there is no point
spectrum at the threshold $v_{c,1}+2$ (see Proposition \ref{propfeb1}). 
Moreover, there exists a second critical value $v_{c,2}> v_{c,1} $ such that if $v\in [v_{c,1},v_{c,2}]$ the spectrum of 
$h(v)$ is purely absolutely continuous, but if $v>v_{c,2}$, an eigenvalue $\lambda(v)$ appears in the interval $(2, v-2)$ and stays there. 

In order to model the switching procedure, we need to make $v$ time dependent. We will only consider switching procedures for which $v(s)=0$ if $s\leq -1$. Moreover, we make the following assumption:
\begin{assumption}
$v$ is twice piecewise differentiable on $(-1,0)$ with uniformly bounded second derivative, and $v$ is continuous at $-1$ and $0$.
\end{assumption}
We are interested in situations where the discrete spectra of $h(v(-1))$ and $h(v(0))$ consist of precisely one eigenvalue. We will only consider the following three generic situations.
\begin{enumerate}
\item The first case is when the potential bias $v$ is $\mathcal{ C}^2$ on $[-1,0]$ and does not cross the critical
values, such that the discrete instantaneous eigenvalue is always present. See Figure~\ref{f0}.
%\begin{figure}\centering
%\includegraphics[scale=0.8]{figure0.pdf}
%\caption{The first situation}\label{f0}
%\end{figure}
\item The second situation is when the bias potential crosses twice the critical value $v_{c,1}$, causing the instantaneous eigenvalue
$\lambda(v(s))$ to disappear at some point in time and then to reappear at a later moment.  We model this potential bias by a function
$v:[-1,0]\to [0, v_{c,2})$, continuous at $-1$ and $0$, with $v(-1)=0$ and $v(0)= v_{c,1}-1$. We assume that
there exist $-1<s_{c}<s'_{c}<0$ and  $0<\delta<<1$ such that: 
\begin{align}\label{mars6}
v(s_{c}-0)&= v(s'_{c}+0)=v_{c,1}-\delta, \nonumber \\
 v(s_{c})&=v(s_{c}+0)=v_{c,1}+\delta=v(s'_{c})=v(s'_{c}-0),\nonumber \\
 v_{c,1}+\delta&\leq v(s)<  v_{c,2},\; s\in [s_{c}, s'_{c}];\nonumber \\
v(s)&\leq v_{c,1}-\delta,\; s\not\in [s_{c}, s'_{c}].
\end{align}
See Figure~\ref{f1}.
%\begin{figure}\centering
%\includegraphics[scale=0.8]{figure1.pdf}
%\caption{The second situation}\label{f1}
%\end{figure}

Keeping in mind that the coupling constant $\tau$ in $h_T$ must be small, 
this potential bias insures that the instantaneous Hamiltonian
$h(v(s))$ will have exactly one discrete eigenvalue if $s\in
[-1,s_{c})\cup (s'_{c},0]$, and purely absolutely continuous
spectrum for $s\in [s_{c},s'_{c}]$. At the end we will let $\delta$ go to
zero, but until then the potential bias has a small discontinuous jump at $s_{c}$ and $s'_{c}$. 

\item The third physically interesting switching procedure (from the mathematical point of view being though closely related to the second situation) 
is the one in which the bias causes the instantaneous eigenvalue $\lambda(v)$ to disappear into the continuous band $[-2+v,2+v]$ for $v_{c,1}\leq v\leq 
v_{c,2}$ 
and to reappear and stay in the interval $(2,v-2)$ for $v>v_{c,2}$. More precisely, we consider 
an increasing function $v:[-1,0]\to [0, v_{c,2}+1]$, continuous at $-1$ and $0$ with $v(-1)=0$ and $v(0)= v_{c,2}+1$. 
We assume that there exist $-1<s_{c}<s'_{c}<0$ and  $0<\delta<<1$ such that: 
\begin{align}\label{mars60}
&v(s_{c}-0)=v_{c,1}-\delta,\qquad\quad v(s'_{c}+0)=
v_{c,2}+\delta, \nonumber \\
 &v(s_{c})=v(s_{c}+0)=v_{c,1}+\delta,\quad v(s'_{c})=v(s_{c}'-0)
=v_{c,2}-\delta.
\end{align}
See Figure~\ref{f2}.
%\begin{figure}\centering
%\includegraphics[scale=0.8]{figure2.pdf}
%\caption{The third situation}\label{f2}
%\end{figure}
\end{enumerate}
\begin{figure}\centering
\includegraphics[scale=0.8]{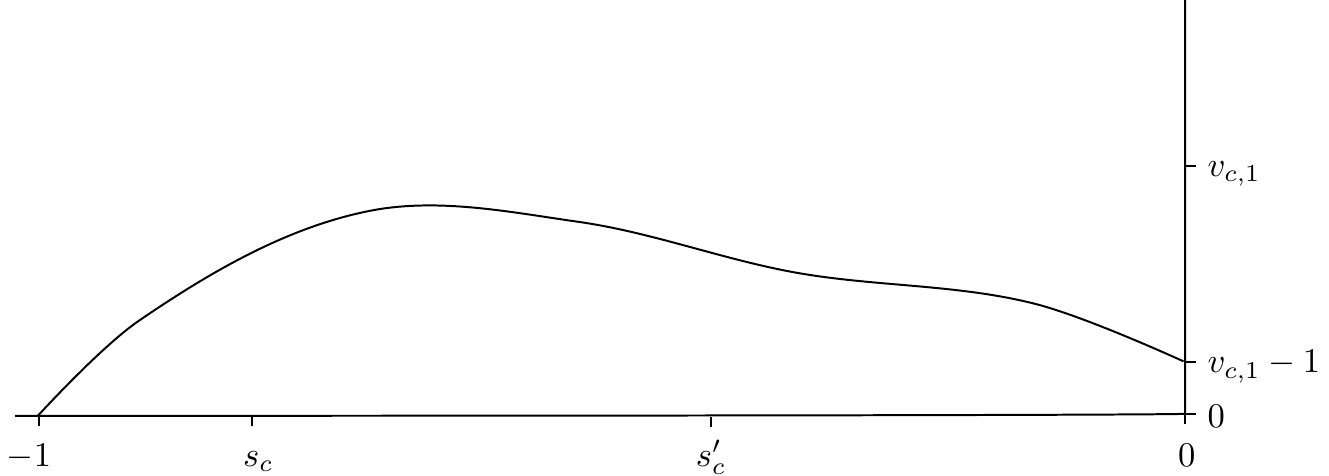}
\caption{The first situation}\label{f0}
\end{figure}
\begin{figure}\centering
\includegraphics[scale=0.8]{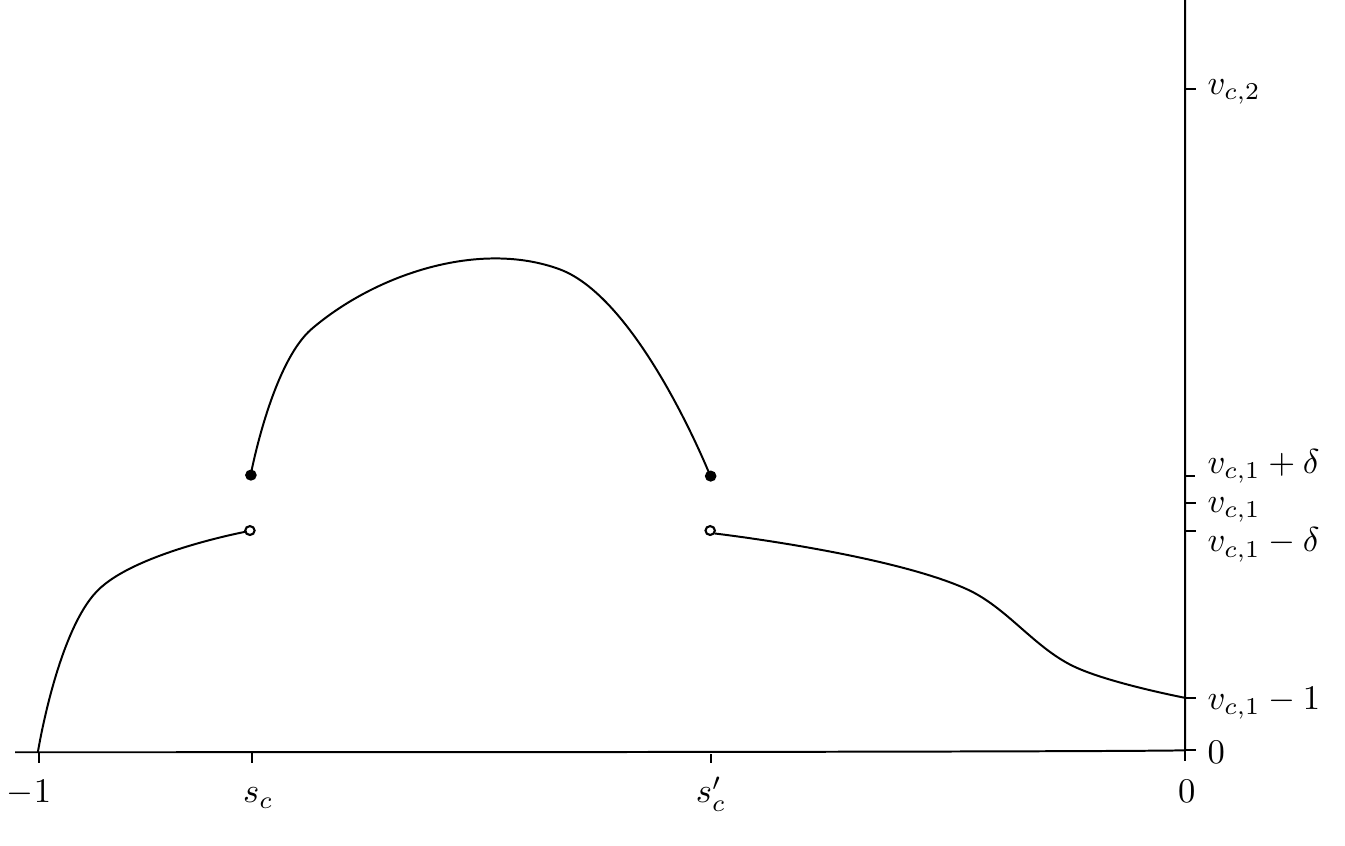}
\caption{The second situation}\label{f1}
\end{figure}
\begin{figure}\centering
\includegraphics[scale=0.8]{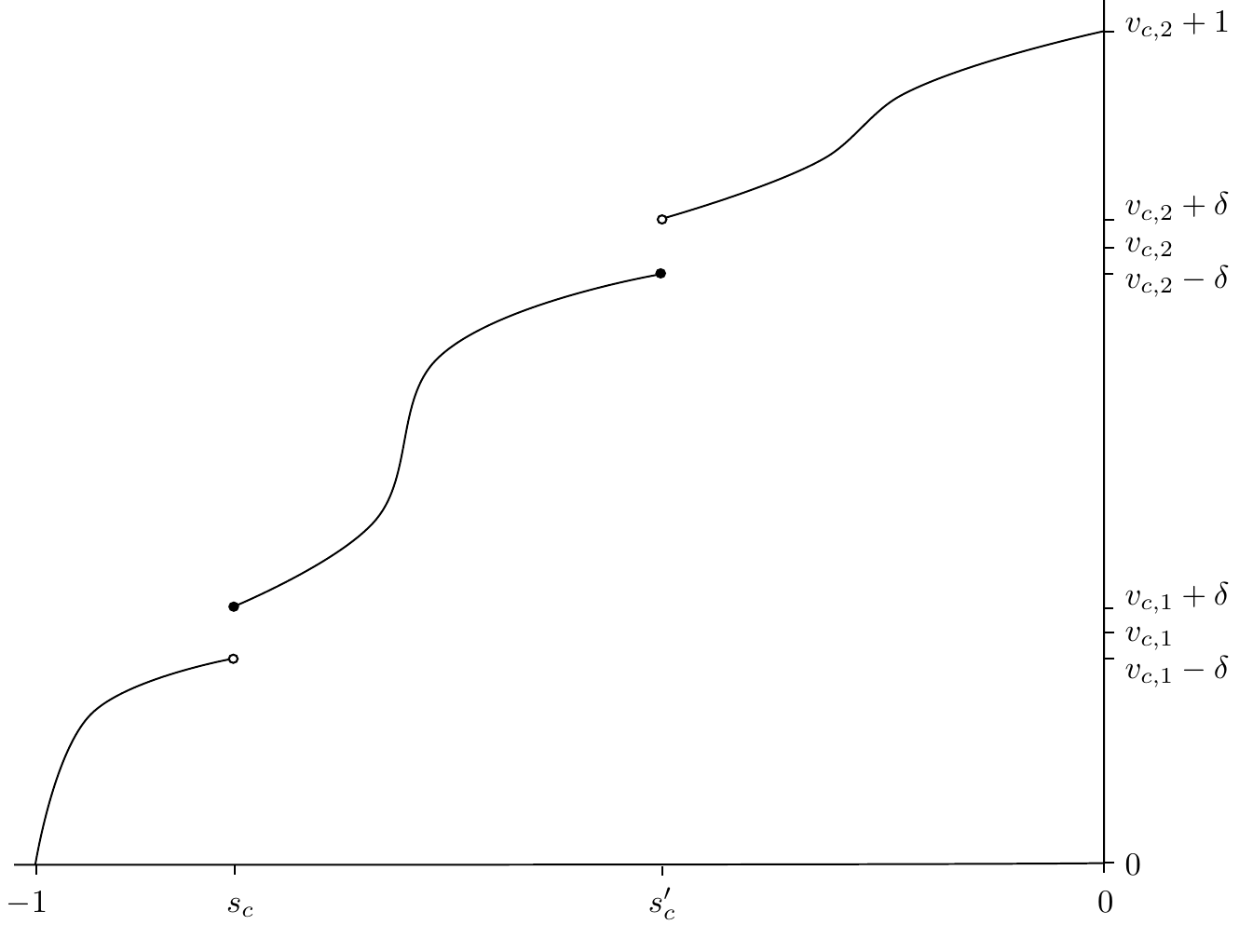}
\caption{The third situation}\label{f2}
\end{figure}

 In what follows, we will replace the notation  $h(v(s))$ by $h(s)$ at all points where $v$ is continuous, and also write $h(s\pm 0):= h(v(s\pm 0))$ for the (finite) set of points $s$ at which $v$ has discontinuities. Moreover, we will adopt the same convention for various functions of $h(s)$ e.g.:
\begin{equation}
 P_d(s)=P_d(h(s))= |\psi (s)\rangle \langle   \psi(s) |  ,\;\; P_{ac}(s)=P_{ac}(h(s)),
\end{equation}
where $\psi(s)$ is an eigenstate corresponding to the (only one) eigenvalues of $h(s)$. We recall that
$P_{ac}(h_L)=\Pi_- +\Pi_+$. Let us point out that the wave  operators 
\begin{equation}\label{mf2}
\Omega(s):=\slim_{t\to -\infty}e^{it h(s)}e^{-it h_L(s)}P_{ac}(h_L)
\end{equation}
exist and are asymptotically complete \cite{Yafaev}.

\subsection{The results}

 If $\eta>0$ is a small adiabatic parameter, we consider the time dependent Hamiltonian
$$h(v(\eta t))=h_S+h_L+v(\eta t)\Pi_-+h_T=h_0(v(\eta t))+h_T.$$
We denote by $U(t,t_0)$ the unitary solution of the time dependent
Schr\"odinger equation
$$iU'(t,t_0)=h(\eta t)U(t,t_0),\quad U(t_0,t_0)=1,\quad
-\eta^{-1}<t,t_0<0.$$ 
The initial equilibrium state of the system is characterized by a
density matrix operator 
$\rho_{\rm eq}$ which is assumed to be a function of the operator
$h(-1)=h_0+h_T$, i.e. 
$\rho_{\rm eq}=f_{\rm eq}(h(-1))$. An example could be $f_{\rm
  eq}(x)={1}/{(e^{\beta(x-\mu)}+1)}$ with $\beta>0$ and
$\mu\in\mathbf{R}$. For physical reasons we demand that $0<\|f_{\rm
  eq}\|_\infty\leq 1$. Note that the one-particle density matrix is not a trace class operator. Given a trace class observable $A$, its equilibrium expectation is given by 
  $\Tr \{\rho_{\rm eq}A \}$. 
  
 In the remote past $t<-\eta^{-1}$, the density matrix operator equals $\rho_{\rm eq}$ and it is time
independent. At time $-\eta^{-1}$ we start out the potential 
bias $v$ and we let it evolve. 

The density matrix operator solves the Liouville equation $i\rho_{\eta}'(t)=[h(\eta t), \rho_{\eta}(t)]$ on the interval 
$-\eta^{-1}< t< 0$, and it is given by the formula 
$$\rho_{\eta}(t):=U(t,-\eta^{-1})\rho_{\rm
  eq}U(t,-\eta^{-1})^*.$$
  
  At $t=0$, our state will be described by:
$$\rho_{\eta}:=\rho_{\eta}(0)=U(0,-\eta^{-1})\rho_{\rm
  eq}U(0,-\eta^{-1})^*=U(-\eta^{-1},0)^*\rho_{\rm
  eq}U(-\eta^{-1},0),$$ 
where in the last equality we used the fact that $U(t,t')^*=U(t',t)$.

If $A\in B_1(\mathcal{H})$ is a given self-adjoint trace class observable, 
the question we want to answer in all three cases is the existence of the following adiabatic limit:
\begin{equation}\label{mars10}
 \langle A\rangle:=\lim_{\eta \searrow 0}\; \Tr \{\rho_{\eta}A \}=
\lim_{\eta \searrow 0}\; \Tr \{U(-\eta^{-1},0)^*f_{\rm eq}(h)U(-\eta^{-1},0)A \}.
\end{equation}
If $A= P_d(0)=|\psi(0)\rangle \langle \psi(0)|$, then $\langle A\rangle$ represents the probability of arriving at the final bound state at the end of the switching procedure. Now 
here is our main result:

\begin{theorem}\label{mainth} We have the following situations:

 \noindent {\rm (i)}  If the potential bias $v$ is chosen in such a way that the
instantaneous eigenvalue $\lambda(s)$ of $h(s)$ is always bounded away from the instantaneous continuous spectrum and stays in the interval $(v(s)+2,\infty)$, then:
\begin{equation}\label{mars3411}
\langle A\rangle =\Tr\{\Omega (0)f_{\rm eq}(h_L)P_{{\rm ac}}(h_L)\Omega (0)^*A \}
+\langle \psi(0)\,|\,A\psi(0)\rangle\; f_{\rm eq}(\lambda(-1)).
\end{equation}
If $A= P_d(0)=|\psi(0)\rangle \langle \psi(0)|$, then $\langle P_d(0)\rangle =f_{\rm eq}(\lambda(-1))$.
 
\noindent {\rm (ii)}
If the potential bias $v$ satisfies either \eqref{mars6} or \eqref{mars60} then
\begin{align}\label{mars341}
\langle A\rangle&=\Tr\{\Omega (0)f_{\rm eq}(h_L)P_{{\rm ac}}(h_L)\Omega (0)^* A \} +\langle \psi(0)\,|\,A\psi(0)\rangle \nonumber\\
&\quad{}\cdot  \langle \psi (s'_c+0)\,|\,
  \Omega (s'_c-0) f_{\rm eq}(h_L)P_{{\rm ac}}(h_L)\Omega (s'_c-0)^*\psi (s'_c+0)\rangle.
\end{align}
Moreover, 
\begin{equation}\label{mars40a}
\langle P_{\rm d}(0)\rangle
= 
\langle \psi (s'_c+0)\,|\,
  \Omega (s'_c-0) f_{\rm eq}(h_L)P_{{\rm ac}}(h_L)\Omega (s'_c-0)^*\psi (s'_c+0)\rangle .
\end{equation}
\end{theorem}

\begin{remark}  The first terms in both \eqref{mars3411} and \eqref{mars341} are identical and 
 only depend upon the initial and final values of the bias potential. Moreover, they are the only ones contributing to the adiabatic charge current since in that case $A=i[h(v(0)),\Pi_\pm]$ which implies  $\langle \psi(0)\,|\,A\psi(0)\rangle =0$, thus the second terms disappear.  This is consistent with previous results \cite{CNZ} showing that the adiabatic limit of the current does not depend upon the specific form of the switching. On the contrary, while in ({\rm i}) the second term in the r.h.s. of  \eqref{mars3411}) again only depends upon the initial and final values of the bias potential, in  ({\rm ii}) the second term in the r.h.s. of \eqref{mars341} depends upon the switching procedure but {\em only} via the behavior of the bias potential in the neighborhood of $s'_c$, i.e. at the last passage through a critical value. 
 The fomulas \eqref{mars3411} and \eqref{mars3411} are explicit, but the price is that we assumed a jump of size $2\delta$ for the bias potential at the critical values.
\end{remark}

Our second main result gives $\langle P_{\rm d}(0)\rangle$ for all three situations in the limit $\delta \searrow 0$, i.e. when the discontinuity jump in the bias potential shrinks to zero.
\begin{proposition}\label{propodilo}
\noindent {\rm (i)}
If the bias $v$ varies in such a way that the discrete eigenvalue $\lambda(s)$ of $h(s)$ always remains
separated from the continuous spectrum {\rm(}see Figure \ref{f0}{\rm)}, then:
 \begin{equation}\label{mf1}
\langle P_{\rm d}(0)\rangle=f_{\rm eq}(\lambda(-1)).
 \end{equation}
\noindent {\rm (ii)}
If $\lambda(s)$ enters the continuous spectrum but reappears and stays in the interval $(v(s)+2,\infty)$ {\rm(}see Figure \ref{f1}{\rm)}, we have:
 \begin{equation}
  \lim_{\delta \searrow 0}\langle P_{\rm d}(0)\rangle=f_{\rm eq}(2).
 \end{equation}
\noindent {\rm (ii)}
If $\lambda(s)$ enters the continuous spectrum but reappears and stays in the interval $(2,v(s)-2)$ {\rm(}see Figure \ref{f2}{\rm)}, we have:
 \begin{equation}
  \lim_{\delta \searrow 0}\langle P_{\rm d}(0)\rangle=f_{\rm eq}(-2).
 \end{equation}
\end{proposition}

\begin{remark}
 It might happen that $f_{\rm
  eq}(\lambda (-1))=0$, i.e. the initial state has no discrete component. If the continuous spectrum is not crossed, the 
probability of finding the system in the final discrete state is still zero (see \eqref{mf1}). But if either $f_{\rm
  eq}(2)=1$ or $f_{\rm
  eq}(-2)=1$, by entering the continuous band this probability can be made equal to $1$. 

One might think that Proposition \ref{propodilo} (ii) and (iii) also cover the generic (i.e. without jumps) bias potentials. Unfortunately this is not  true since we {\em first} took the adiabatic limit $\eta \searrow 0$, and after that $\delta \searrow 0$. In order to cover the general case one has to perform the limits in the reversed order. As it is known from the spontaneous pair creation case \cite{Nenciu2,Nenciu3,PD} this is a hard technical problem (even more demanding for the problem at hand) since one has to control the evolution in the adiabatic limit near the critical times and this relies on detailed (model dependent) spectral and scattering study near criticality.  Notice that contrary to the spontaneous pair creation problem, in our case one has to deal with a critical Hamiltonian having a resonance at the threshold, and not a bound state.  We believe that a careful study of spectral and scattering theory near energetic thresholds enlarging and streamlining the results in \cite{P,PD,JN,RU} is needed in order to solve this very interesting open problem.

We can allow multiple crossings of both critical values, and the results are rather similar. The 
probability of finding the system on the final discrete state after taking the limit  $\delta\searrow 0$ will be given by $f_{\rm
  eq}(\pm 2)$, depending on which critical value was the last one crossed. 
\end{remark}

\subsection{The contents of the paper}

After this introductory section which also included the main results of our paper, we continue in Section \ref{appendix} with a detailed spectral analysis 
of the operator $h(v)$. The results are rather straigthforward but are needed in order to study the behavior of the discrete eigenvalue around 
the critical values of the bias. In Propositions \ref{propfeb1} and \ref{febprop2} we show that the eigenvector 
corresponding to this discrete eigenvalue becomes more and more delocalized in the leads as we get closer and closer to the critical values. 

In Section \ref{section3} we prove Theorem \ref{mainth}. The main technical tool can be found in Proposition \ref{propestim}, which in some sense can 
be considered to be an adiabatic theorem for the continuous spectrum (see \cite{AEGS,NT,MN,D} for   results about adiabatic limit in scattering theory). 

In Section \ref{section4} we prove Proposition \ref{propodilo}. More precisely, we investigate the limit $\delta\searrow  0$ of the probability of finding 
the system in the final bound state. The fact that the discrete eigenvector becomes very delocalized when we approach the critical values of the bias plays 
 a crucial role in the proof.

\section{Spectral analysis of $h(v)$}\label{appendix}

Introduce the notation $R(z)$ for the inverse of $h(v)-z$ in ${\mathcal H}$, and 
$R_L(z)$ for the inverse of $h_L+v\Pi_--z$ in ${\mathcal H}_L$, extended by $0$ on ${\mathcal H}_S$, i.e. 
$$R_L(z)=0\oplus (h_-+v-z)^{-1}\oplus (h_+-z)^{-1}.$$
Note that the $v$ dependence is not made explicit in this notation.

Let us split the tunneling 
Hamiltonian $h_T$ as follows:
\begin{equation}
h_T=h_{LS}+h_{SL}, \qquad h_{LS}=\tau |0_{-}\rangle\langle S|+\tau |0_{+}\rangle\langle S|=h_{SL}^*.
\end{equation}
Define the effective Hamiltonian as $h_{{\rm eff}}(z):=(E_0-z)|S \rangle \langle S|-h_{SL}R_L(z)h_{LS}$; when restricted 
to ${\mathcal H}_S$, the effective Hamiltonian is 
$h_{{\rm eff}}(z)=G(z;v)|S \rangle \langle S|$, where
\begin{equation}\label{Heff}
G(z;v):=E_0-z-\tau^2\langle 0_-\,|\, (h_-+v-z)^{-1} 0_-\rangle -
\tau^2 \langle 0_+\,|\, (h_+-z)^{-1} 0_+\rangle.
\end{equation} 
Then the Feshbach formula for $R(z)$ reads (see e.g \cite{CJM1}):
\begin{equation}\label{feshchc}
R(z)=R_L(z)+\frac{1}{G(z;v)}\{1-R_L(z)h_{LS}\}|S\rangle\langle S|
\{1-h_{SL}R_L(z)\}.
\end{equation}
Let us denote by $r_\pm(z)=(h_\pm-z)^{-1}$ the resolvents of the leads extended by $0$ outside $l^2(\mathbf{N}_\pm)$ 
in the obvious way. With this notation \eqref{feshchc} writes:
\begin{align}\label{feb2}
R(z)&=R_L(z)+\frac{1}{G(z;v)}|S\rangle\langle S| 
-\frac{\tau}{G(z;v)}|S\rangle \langle 0_-| r_-(z-v)\nonumber\\ &\quad-\frac{\tau}{G(z;v)}|S\rangle \langle 0_+|r_+(z)
+\frac{\tau^2}{G(z;v)}\Bigl[-\frac{1}{\tau}r_-(z-v) |0_-\rangle \langle S| \nonumber\\
&\qquad+r_-(z-v)|0_-\rangle 
\langle 0_-| r_-(z-v)
+ r_-(z-v)|0_-\rangle 
\langle 0_+| r_+(z)\Bigr]\nonumber\\ 
&\quad+\frac{\tau^2}{G(z;v)}\Bigl[-\frac{1}{\tau}r_+(z) |0_+\rangle \langle S| +r_+(z)|0_+\rangle 
\langle 0_-| r_-(z-v)\nonumber \\
&\qquad+ r_+(z)|0_+\rangle 
\langle 0_+| r_+(z)\Bigr].
\end{align}
What is needed in \eqref{Heff} and \eqref{feb2} are the formulas for $r_\pm(z)$. We recall them from \cite{CJM1}.

Let
\begin{equation}\label{matqa2}
\zeta_1(z)=\frac{z}{2}\bigl(1-\sqrt{1-4/z^2}\bigr).
\end{equation}
where $\sqrt{x}$ is chosen to have the cut $(-\infty,0)$ i.e. $\sqrt{x}:=\sqrt{|x|}e^{(i/2)\arg(x)}$,
$\arg(x)\in
(-\pi,\pi)$. Notice that $\zeta_1(z)$ is real analytic in $\mathbf{C}\setminus [-2,2]$ and vanishes at infinity like $1/z$. Also, $\zeta_1$ is decreasing 
on the separate intervals $(-\infty,-2)$ and $(2,\infty)$, while $\zeta_1(\pm 2)=\pm 1$. 

As it was shown in \cite{CJM1} we have: 
\begin{align}\label{matrellead}
\langle 0_\pm,(h_\pm -z)^{-1}0_\pm \rangle &=-\zeta_1(z),\nonumber \\
\langle m_\pm,(h_\pm -z)^{-1}n_\pm \rangle &=
\frac{\zeta_1(z)}{\zeta_1^2(z)-1}\bigl(\zeta_1(z)^{|m-n|}
-\zeta_1(z)^{m+n+2}     \bigr),\; n,m\geq 0.
\end{align}

In particular \eqref{Heff} becomes:
\begin{align}\label{feb1}
G(z;v)=E_0-z+\tau^2 \{\zeta_1(z-v)+\zeta_1(z)\}.
\end{align}

\subsection{The point spectrum of $h(v)$ for $v\geq 0$} 

As it is well known from the Feshbach lemma, the discrete spectrum of the coupled operator $h(v)$ is given by the real solutions of the 
equation $G(x;v)=0$, where 
$x\not\in[-2,2]\cup[-2+v,2+v]$. In the introduction we announced that we are interested 
in the case when $E_0$ is large and far from the unbiased continuous spectrum, 
while the coupling $\tau$ is weak. Let us assume that 
 $E_0\geq 10$ and $0<|\tau|<<1$. 
 
 For any fixed $v\geq 0$, by differentiating in \eqref{Heff} we obtain:
\begin{equation}\label{feb3}
\partial_xG(x;v)\leq -1,\quad  x\not\in[-2,2]\cup[-2+v,2+v].
\end{equation}
which means that the map $G(\cdot;v)$ is strictly decreasing on the intervals $(-\infty,-2)$ and 
$(2+v,\infty)$. If $v>4$ then $G(\cdot;v)$ is also decreasing on $(2,-2+v)$. In addition, 
\begin{equation}\label{Ginf}
\lim_{x\to \pm \infty}G(x;v)=\mp \infty.
\end{equation}
We first show that there are no eigenvalues in the interval $(-\infty,-2)$. Indeed,  using \eqref{Ginf} and the 
inequality $G(-2;v)\geq  E_0+2-\tau^2 \sup_{\nu \geq 0}|\zeta_1(-2-\nu)+\zeta_1(-2)|>0$ (remember that  
$\tau$ is small enough and $E_0\geq 10$), we conclude that $G$ does not change sign on $(-\infty,-2)$, 
hence there are no discrete eigenvalues there. 

Now consider the interval $(2+v,\infty)$. Because $G(\cdot,v)$ is strictly decreasing, we have that:
$$ -\infty<G(x;v)\leq G(2+v;v),\quad 2+v\leq x<\infty.$$ For all $v\geq 0$, formula \eqref{feb1} gives: 
\begin{align}\label{feb4}
G(2+v;v)=E_0-2-v+\tau^2 \{\zeta_1(2)+\zeta_1(v+2)\}.
\end{align}
There exists a (unique) eigenvalue if and only if $G(2+v;v)>0$, thus we need to investigate $G(2+v;v)$ as a 
function of $v$.
 
Because $\zeta_1$ is decreasing on $(2,\infty)$ we have that $G(2+v;v)$ is strictly decreasing with $v$. 
If $0\leq v\leq 4$ then 
$G(2+v;v)\geq 4-\tau^2 \sup_{\nu \geq 0}|\zeta_1(2)+\zeta_1(\nu +2)|>0$ for small enough 
$\tau$. It means that there exists exactly one non-degenerate eigenvalue 
$\lambda(v)\in (v+2,\infty)$ if $0\leq v\leq 4$. If $v$ increases even more, then $G(2+v;v)$ decreases 
approaching zero. 
There will be a 
critical value $v_{c,1}>4$ such that $G(2+v_{c,1};v_{c,1})=0$ and $G(2+v;v)<0$ if $v> v_{c,1}$. The value of 
$v_{c,1}$ must be close to $E_0-2$ if $\tau$ is small. We conclude that 
if $v> v_{c,1}$ there is no discrete spectrum on $(v+2,\infty)$. 

Third, we need to investigate what happens in the interval $(2,-2+v)$ when $v>4$. We see that $G(2;v)>0$ 
for all $v>4$, and we have the inequality:
$$ G(-2+v;v)\leq G(x;v)\leq G(2;v),\quad 2\leq x\leq -2+v.$$
Thus we only need to investigate the sign of 
$$G(-2+v;v)=E_0+2-v+\tau^2 \{\zeta_1(-2)+\zeta_1(-2+v)\},\quad v>4.$$
Note that  $G(-2+v;v)$ strictly decreases with $v$. If $\tau$ is small 
enough, then $G(-2+v;v)>G(2+v;v)\geq 0$ for all $4\leq v\leq  v_{c,1}$, which means that there is no discrete spectrum  in the interval 
$(2,-2+v)$ if $4< v \leq v_{c,1}$. If $v>v_{c,1}$ then $G(-2+v;v)$ continues to decrease until it reaches zero and 
afterwards it becomes negative. This defines a second critical value $v_{c,2}\approx E_0+2$ such that 
$0=G(-2+v_{c,2};v_{c,2})>G(-2+v;v)$ 
for every $v>v_{c,2}$. This generates a non-degenerate eigenvalue $\lambda(v)$ in the interval $(2,-2+v)$. 

To summarize: if $v\in [0, v_{c,1})$ there exists a unique non-degenerate eigenvalue in the interval $(v+2,\infty)$. If 
$v_{c,1}<v<v_{c,2}$ there are no discrete eigenvalues, and if $v_{c,2}<v$ then we again have a unique 
non-degenerate eigenvalue in the interval $(2,-2+v)$. 

There are no embedded eigenvalues in the set $(-2,2)\cup (-2+v,2+v)$, for all $v\geq 0$; the 
explanation is that the imaginary part of $G(x+i0_+;v)$ is not zero if $x\in (-2,2)\cup (-2+v,2+v)$, see \eqref{feb1} 
and \eqref{matqa2}. The only remaining situation where eigenvalues could exist is at thresholds, i.e. 
when $v$ is either $v_{c,1}$ or $v_{c,2}$. But we will show in the next subsection that this is not the case. 

\subsection{Behavior near thresholds}\label{sect22}

Let us assume that $v_{c,1}-v>0$ is very small, which means that the eigenvalue 
$\lambda(v)\in (v+2,\infty)$ still exists but $0<\lambda(v) -(v+2)<<1$. From \eqref{feb1} and 
\eqref{matqa2} we obtain:
\begin{align}\label{feb14}
F(t,v)&:=G(t^2+v+2;v),\quad G(x;v)=F(\sqrt{x-v-2},v),\\
F(t,v)&=E_0-v-2-t^2+\tfrac{1}{2}\tau^2\bigl(2+t^2- t\sqrt{4+t^2}+t^2+v+2\nonumber \\
&\quad-\sqrt{(v+t^2)(v+4+t^2)}\bigr),\nonumber 
\end{align}
which admits a smooth extension near the point $(0,v_{c,1})$, 
with $F(0,v_{c,1})=0$ and $\partial_t F(0,v_{c,1})=-\tau^2\neq 0$. Then the implicit function theorem gives us a smooth 
map $t(v)$ defined in a neighborhood of $v_{c,1}$ where $F(t(v),v)=0$. Since  
$\partial_vF(0,v_{c,1})=-1+\mathcal{O}(\tau^2)\neq 0$ if $\tau$ is
small enough, we have $t(v)\sim v_{c,1}-v>0$. It follows that 
\begin{equation}\label{mars1}
\lambda(v)=t^2(v)+v+2,
\end{equation}
and $\lambda(v)-(v+2)\sim (v_{c,1}-v)^2$ near the threshold. 

Moreover, since $G(x;v)=F(\sqrt{x-v-2},v)$ we have: 
\begin{align}\label{feb10}
\partial_xG(x;v)\sim
-\frac{\tau^2}{2\sqrt{x-(v+2)}},\quad 0<x-(v+2)<<1, 
\end{align}
and this estimate holds for $v$ in a small neighborhood of
$v_{c,1}$. In particular, for $v=v_{c,1}$ and using
that $G(v_{c,1}+2;v_{c,1})=0$ we have after integration:
\begin{align}\label{feb11}
G(x;v_{c,1})\sim
-\tau^2\sqrt{x-(v_{c,1}+2)},\quad 0<x-(v_{c,1}+2)<<1. 
\end{align}
\begin{proposition}\label{propfeb1}
If $v$ is either $v_{c,1}$ or $v_{c,2}$, then neither
$v_{c,1}+2$ nor $v_{c,2}-2$ are eigenvalues of $h$.
\end{proposition}
\begin{proof}
We only give the proof for $v_{c,1}$; the other case is similar. We
will prove that for every basis vector 
$\psi\in \{|m_\pm\rangle:\; m\geq
0\}\cup \{|S\rangle\}$ we have
\begin{equation}\label{feb111}
\lim_{x \searrow v_{c,1}+2}(x-v_{c,1}-2)\langle \psi|
R(x)\psi\rangle =0.
\end{equation}
Coupling this with the fact that we always have  
$$\lim_{x \searrow v_{c,1}+2}(x-v_{c,1}-2)\langle \psi| R(x)P_{\rm ac}\psi\rangle =0,$$
it shows that if there exists some orthogonal projection $P$ such
that $R(x)=R(x)P_{\rm ac}+\frac{1}{x-v_{c,1}-2}P$, then necessarily
$\langle \psi| P\psi\rangle=\|P\psi\|^2=0$ for all basis elements,
hence $P=0$. 
Let us prove \eqref{feb111} for $\psi=|m_-\rangle$. From \eqref{feb2} we have:
\begin{align}\label{feb12}
\langle m_-|&R(x)m_-\rangle\nonumber\\
&=\langle m_-|r_-(x-v_{c,1})m_-\rangle
+\frac{\tau^2}{G(x;v_{c,1})} |\langle m_-|r_-(x-v_{c,1})
  0_-\rangle|^2,
\end{align}
and from \eqref{matrellead}:
\begin{align}\label{feb13}
\langle m_-|R(x)m_-\rangle
&=-\zeta_1(x-v_{c,1})\frac{1-(\zeta_1^2(x-v_{c,1}))^{m+1}}{1-\zeta_1^2(x-v_{c,1})}\nonumber\\
&\quad+\frac{\tau^2}{G(x;v_{c,1})}|\zeta_1(x-v_{c,1})|^{2m+2}.
\end{align}
Since $\zeta_1^2(x-v_{c,1})$ goes to $1$ when $x$ converges to
$v_{c,1}+2$, it follow that the only singular behavior comes from
\eqref{feb11}, thus:
$$(x-v_{c,1}-2)\langle m_-|R(x)m_-\rangle\sim
\sqrt{x-v_{c,1}-2} \to 0.$$
The other cases are similar and we do not treat them.
\end{proof}

\begin{proposition}\label{febprop2}
The one dimensional projection $P_{\rm d}(v)$ associated with
$\lambda(v)$ is differentiable on $(0,v_{c,1})\cup
(v_{c,2},\infty)$, and for every compact $K\subset\mathbf{R}$ there exists a
constant $C_K>0$ such that we have the (optimal) bound:
$$\|P_{\rm d}'(v)\|\leq \frac{C_K}{(v-v_{c,2})(v_{c,1}-v)},\quad
v\in K\cap [(0,v_{c,1})\cup
(v_{c,2},\infty)].$$
\end{proposition}
\begin{proof} We only concentrate on the case in which $0<v_{c,1}-v<<1$, that is
just before the eigenvalue $\lambda(v)\in (v+2,\infty)$ hits the first 
threshold and disappears. 
The eigenvalue $\lambda(v)=t^2(v)+v+2$ is smooth as a function of $v$
because $t(v)$ is, see the discussion preceding \eqref{mars1}. Using the Riesz formula for the projection
$P_{\rm d}(v)$ associated with $\lambda(v)$, we see from \eqref{feb2}
that it will consist of a finite sum of rank one operators, like for example 
$$
 P_a(v):=-\frac{\tau^2}{(\partial_xG)(\lambda(v);v)}r_-(2+t^2(v))|0_-\rangle
\langle 0_-|r_-(2+t^2(v)).
$$
The above operator turns out to be the most singular when $v$ lies
near the critical values. From \eqref{feb14} we see that
$(\partial_xG)(\lambda(v);v)=
\frac{(\partial_tF)(t(v),v)}{2t(v)}$ hence:
\begin{equation}\label{mars0}
P_a(v)=-\frac{2\tau^2 t(v)}{(\partial_tF)(t(v),v)}r_-(2+t^2(v))|0_-\rangle
\langle 0_-|r_-(2+t^2(v)).
\end{equation}
If $k\geq 1$ and $|x|>2$ we have:
\begin{align}\label{mars2}
\Vert r_\pm^k(x)|0_\pm\rangle\Vert&=\sqrt{\langle
  0_\pm|r_\pm^{2k}(x) 0_\pm\rangle}\nonumber \\
  &=\frac{1}{\sqrt{(2k-1)!}}\sqrt{ \partial_x^{2k-1}\langle
  0_\pm|r_\pm(x) 0_\pm\rangle}=\frac{\sqrt{ \zeta_1^{(2k-1)}(x)}}{\sqrt{(2k-1)!}}.
\end{align}
From \eqref{matqa2} we see that 
$$\zeta_1^{(2k-1)}(x)\sim \frac{1}{(x^2-4)^{2k-3/2}},\quad
0<|x|-2<<1,\quad k\geq 1,$$
which implies:
\begin{align}\label{mars3}
 \Vert r_\pm^k(x)|0_\pm\rangle \Vert\sim \frac{1}{(x^2-4)^{k-3/4}},\quad
0<|x|-2<<1,\quad k\geq 1.
\end{align}
In particular:
\begin{align}\label{mars4}
\Vert r_-^k(2+t^2(v))|0_\pm\rangle\Vert&\sim
\frac{1}{t(v)^{2k-3/2}}\nonumber\\
&\sim \frac{1}{(v_{c,1}-v)^{2k-3/2}},\;
0<v_{c,1}-v\ll 1,\; k\geq 1.
\end{align}
If $k=1$ we obtain that $  \Vert
  r_-(2+t^2(v))|0_\pm\rangle \Vert^2\sim 1/t(v)$, which shows
that $P_a(v)$ is bounded near the threshold. Now by differentiating
\eqref{mars0} with respect to $v$, and keeping in mind that $t'(v)$ is
bounded, we see that the singular behavior is given by:
\begin{align*}\max\bigl \{ t(v)^2    \Vert
  r_-^2(2+t^2(v))|0_\pm\rangle \Vert   \Vert
  r_-(2+t^2(v))|0_\pm\rangle \Vert,\\  \Vert
  r_-(2+t^2(v))|0_\pm\rangle \Vert^2 \bigr\}\sim
\frac{1}{t(v)},
\end{align*}
which is of the type claimed by the proposition. 

Now we can analyze all the other terms given by the Riesz formula, and
notice that they contain at most one resolvent $r_-(2+t^2(v))$, which
is the only object which produces singularities. Since 
$P_a(v)$ contains two such resolvents, it is the most singular
object. Clearly, for $v>v_{c,2}$ the singular object would be
$r_-(-2-t^2(v))$ and $P_a(v)$ is the term containing two such resolvents. It
turns out that 
near both critical potentials $v_c$ we have: 
\begin{align}\label{mars5}
\|P_{\rm d}(v)-P_a(v)\|&=\mathcal{O}(\sqrt{|v-v_c|}),\\ \|P_{\rm d}'(v)-P_a'(v)\|&=\mathcal{O}(|v-v_c|^{-1/2}).
\end{align}
We now prove that $P_d(v)$ converges strongly to zero when $v\to
v_{c,1}$. First, if $m\geq 0$ is fixed, then from \eqref{mars0} and
\eqref{matrellead} we have: 
$$\langle m_-|P_a(v) m_-\rangle=-\frac{2\tau^2
  t(v)}{(\partial_tF)(t(v),v)}\zeta_1^{2m+2}(2+t(v)^2)\sim t(v). $$
 Second, using  $P_d(v)^2=P_d(v)$ and \eqref{mars5} we have:
  
    $$\Vert P_d(v)m_-\Vert ^2\leq \|P_{\rm d}(v)-P_a(v)\|+\langle m_-|P_a(v) m_-\rangle \sim \sqrt{v_{c,1}-v}. $$
    
  Up to an $\varepsilon/2$ argument one can now show that $P_d(v)$ converges strongly to zero. 
All this is consistent with the fact that there are no eigenvalues at
thresholds when $v=v_c$. The projections $P_{\rm d}(v)$ delocalise more
and more and converge strongly to zero.  Notice that this implies that if $P_d(v)\psi(v)=\psi(v)$ then $\psi(v)$ converges weakly to zero.
\end{proof}

\section{Proof of Theorem \ref{mainth}}\label{section3}

We only prove in detail \eqref{mars341}, the other identity follows after a similar argument.

\subsection{A general propagation estimate}
Let $s_{-1}<s_0\leq 0$ and consider a function
$\nu:(-\infty,s_0]\mapsto \mathbf{R}$ such that $\nu(s)=\nu(s_{-1})$ for $s\leq s_{-1}$, 
and $\nu$ is piecewise continuous, in particular continuous at $s_0$. The values $\nu(s_{-1})$ and $\nu(s_{0})$ are not 
one of the two critical values $v_{c,1}$ and $v_{c,2}$. Denote with $P_{{\rm ac, i}}$ (respectively $P_{{\rm d, i}}$) 
the projection on the absolutely continuous (respectively discrete) subspace of 
$h+\nu(s_{-1})\Pi_-$, and with $P_{{\rm ac, f}}$ (respectively $P_{{\rm d, f}}$) the projection on the absolutely continuous 
(respectively discrete) subspace of $h+\nu(s_{0})\Pi_-$.

Remember that $h_0=h_S+h_L$ and $h=h_0+h_T$. The subspace of absolute continuity of $h_0$ is 
$$P_{{\rm ac}}(h_L)=0\oplus \Pi_-\oplus \Pi_+.$$ 
Denote by 
 \begin{align}\label{mars12}
  \Omega_i :={\rm s}\lim_{t\to -\infty}e^{it (h+\nu(s_{-1})\Pi_-)}e^{-it (h_0+\nu(s_{-1})\Pi_-)}P_{{\rm ac}}(h_L)
 \end{align}
the 'incoming' wave operator between  $h_0+\nu(s_{-1})\Pi_-$ and
$h+\nu(s_{-1})\Pi_-$. This operator exists and is unitary 
between ${\Ran}(P_{{\rm ac}}(h_L))$ and ${\Ran}(P_{{\rm ac, f}})$ \cite{Yafaev,CJM1}. 
In a similar way, define:
\begin{align}\label{mars13}
  \Omega_f:={\rm s}\lim_{t\to -\infty}e^{it (h+\nu(s_{0})\Pi_-)}e^{-it (h_0+\nu(s_{0})\Pi_-)}P_{{\rm ac}}(h_L).
 \end{align}

If $t\leq w\leq s_0/\eta$ we denote by $U(t,w)$ the unitary solution to the evolution equation
\begin{align}\label{mars14}
iU'(t,w)&=h(\eta t)U(t,w),\quad h(s):= h+\nu(s)\Pi_-,\\
U(w,w)&=1,\quad
-\infty <t\leq w.\nonumber
\end{align}

\begin{proposition}\label{propestim}
Let $\Psi\in \mathcal{H}$ be a fixed unit vector. Using the notation $\Delta\nu:=\int_{s_{-1}}^{s_0} \nu(\tau)d\tau$. 
we have the following estimates:
\begin{equation}\label{mars15}
\lim_{\eta\searrow 0}
U\bigl(\frac{s_{-1}}{\eta}, \frac{s_{0}}{\eta}\bigr)^*
\Omega_i  e^{-i\frac{s_{-1}-s_0}{\eta}h_L } 
e^{\frac{i}{\eta}\Delta\nu  \Pi_-}P_{{\rm ac}}(h_L)\Psi
=\Omega_f P_{{\rm ac}}(h_L)\Psi,
\end{equation}
which is equivalent with:
\begin{align}\label{mars16}
 \lim_{\eta\searrow 0}\bigl\Vert 
U\bigl(\frac{s_{-1}}{\eta},&\frac{s_{0}}{\eta}\bigr)\Omega_fP_{{\rm ac}}(h_L)\Psi\nonumber\\
&-\Omega_i
e^{-i\frac{(s_{-1}-s_0)}{\eta}h_L} e^{i\frac{\Delta
    \nu}{\eta}\Pi_-}P_{{\rm ac}}(h_L)\Psi\bigr\Vert=0
\end{align}
and 
\begin{align}\label{mars16'}
 \lim_{\eta\searrow 0}\bigl\Vert 
U\bigl(\frac{s_{-1}}{\eta}, &\frac{s_{0}}{\eta}\bigr)P_{{\rm ac, f}}\Psi
\nonumber\\
&
-\Omega_i
e^{-i\frac{(s_{-1}-s_0)}{\eta}h_L} e^{i\frac{\Delta
    \nu}{\eta}\Pi_-}P_{{\rm ac}}(h_L)\Omega_f^*P_{{\rm ac, f}}\Psi\bigl \Vert=0.
\end{align}
\end{proposition}
\begin{proof}

We start by proving \eqref{mars15}.  Denote by $U_0(t,w)$ the explicit unitary solution of the equation 
\begin{align}\label{mars17}
iU_0'(t,w)&=(h_L+\nu(\eta t)\Pi_-)U_0(t,w),\quad U_0(w,w)=1,\quad
-\infty <t\leq w\nonumber \\
U_0(t,w)&=e^{-i(t-w)h_L}e^{-\frac{i}{\eta}\int_{\eta w}^{\eta t} \nu(s)ds \Pi_-}.
\end{align}
Note that $U_0(t,w)$ commutes with $P_{{\rm ac}}(h_L)$, and equals the identity operator when restricted to the subspace of the sample. 

Choose $f\in {\Ran} P_{{\rm ac}}(h_L)$ such that if $\Phi_E^{\pm}$ is a generalized eigenfunction of $h_L$, 
then the function $\langle \Phi_E^{\pm},f\rangle $ is a smooth function of $E$ with 
a compact support not containing $\pm 2$. First, $\langle S,U_0(t,w)f\rangle =0$. Second, from \eqref{mars17} and by partial integration 
with respect to $E$ it follows that there exists $C_{f,N}>0$ independent of 
$\eta$ such that 
\begin{align}\label{mars18}
\max_{0\leq j\leq N} \{ |\langle j_\pm,U_0(t,w)f\rangle |\}
\leq C_{f,N}\langle t-w\rangle^{-2},\quad \forall t\leq w.
\end{align}
Define
\begin{equation}\label{marsaa}
W(t):=U\bigl(t, \frac{s_{0}}{\eta}\bigr)^*U_0\bigl (t, \frac{s_{0}}{\eta}\bigr )P_{{\rm ac}}(h_L),\qquad t\leq \frac{s_{0}}{\eta}.
\end{equation}

We will compute the strong limit of $W(t)$ when $t\to -\infty$ in two
different ways. Up to an $\varepsilon/2$ argument it is enough to prove
the existence of a strong limit for a function $F$ such that 
$\langle \Phi_E^{\pm},F\rangle $ is a smooth function of $E$ with 
a compact support not containing $\pm 2$.

On one hand, if $t<s_{-1}/\eta$ we have:
\begin{align}\label{mars19}
W(t)F&=U\bigl(t, \frac{s_{0}}{\eta}\bigr)^*U_0\bigl(t, \frac{s_{0}}{\eta}\bigr)P_{{\rm ac}}(h_L)F\nonumber \\
&=
U\bigl(\frac{s_{-1}}{\eta}, \frac{s_{0}}{\eta}\bigr)^*e^{i(t-s_{-1}/\eta)h(\nu(s_{-1}))}
\nonumber\\
&\qquad\cdot e^{-i(t-s_0/\eta)h_L}P_{{\rm ac}}(h_L)e^{-\frac{i}{\eta}\int_{s_0}^{\eta t} \nu(s)ds \Pi_-}F\nonumber\\
&=
U\bigl(\frac{s_{-1}}{\eta}, \frac{s_{0}}{\eta}\bigr)^*e^{i(t-s_{-1}/\eta)h(\nu(s_{-1}))}
e^{-i(t-s_{-1}/\eta)(h_0+\nu(s_{-1})\Pi_-)}P_{{\rm ac}}(h_L)\nonumber\\
&\qquad\cdot e^{-i\frac{s_{-1}-s_0}{\eta}h_L } e^{\frac{i}{\eta}\Delta\nu \Pi_-}F,
\end{align}
where in the second line we used the group property of $U$ and the
fact that $\nu$ is constant at the left of $s_{-1}$. The variable $t$ only appears in the middle, and using 
\eqref{mars13} we obtain:
\begin{align}\label{mars20}
\lim_{t\to -\infty}W(t)F&=
U\bigl(\frac{s_{-1}}{\eta}, \frac{s_{0}}{\eta}\bigr)^*
\Omega_i  e^{-i\frac{s_{-1}-s_0}{\eta}h_L } e^{\frac{i}{\eta}\Delta\nu  \Pi_-}P_{{\rm ac}}(h_L)F.
\end{align}
On the other hand, by differentiating with respect to $t$ the formula
\eqref{marsaa} defining $W(t)F$ and then integrating back 
we obtain:
\begin{align}\label{mars21''}
W(t)F&=P_{{\rm ac}}(h_L)F+i\int_{s_0/\eta}^{t}d\tau 
U\bigl(\tau, \frac{s_{0}}{\eta}\bigr)^*h_TU_0\bigl(\tau, \frac{s_{0}}{\eta}\bigr)
P_{{\rm ac}}(h_L)F\nonumber
 \\
&=P_{{\rm ac}}(h_L)F\nonumber\\
&\quad+i\int_{0}^{t-\frac{s_{0}}{\eta}}d\tau
U\bigl(\tau+\frac{s_{0}}{\eta}, \frac{s_{0}}{\eta}\bigr)^* h_T U_0
\bigl(\tau+\frac{s_{0}}{\eta}, \frac{s_{0}}{\eta}\bigr)P_{{\rm ac}}(h_L)F, 
\end{align}
or using \eqref{mars20}:
\begin{align}\label{mars21}
\lim_{t\to -\infty}&W(t)F=
U\bigl(\frac{s_{-1}}{\eta}\, \frac{s_{0}}{\eta}\bigr)^*
\Omega_i e^{-i\frac{s_{-1}-s_0}{\eta}h_L } e^{\frac{i}{\eta}\Delta\nu \Pi_-}F
= P_{{\rm ac}}(h_L)F\nonumber\\
&+i\int_{0}^{-\infty}d\tau 
U\bigl(\tau+\frac{s_{0}}{\eta}, \frac{s_{0}}{\eta}\bigr)^*h_TU_0
\bigl(\tau+\frac{s_{0}}{\eta}, \frac{s_{0}}{\eta}\bigr)P_{{\rm ac}}(h_L)F.
\end{align}
For fixed $\tau$, we have:
\begin{align*}
U\bigl(\tau&+\frac{s_{0}}{\eta}, \frac{s_{0}}{\eta}\bigr)^*
e^{-i \tau h(\nu(s_0))}\\
&=
1
+i\int_{0}^\tau dt'U\bigl(t'+\frac{s_{0}}{\eta}, \frac{s_{0}}{\eta}\bigr)^*(\nu (\eta t'+s_0)-\nu(s_0))
e^{-i t' h(\nu(s_0))}, 
\end{align*}
which leads to the operator norm estimate:
  \begin{align}\label{mars22'}
 \bigl\Vert U\bigl(\tau+\frac{s_{0}}{\eta}, \frac{s_{0}}{\eta}\bigr)^*
e^{-i \tau h(\nu(s_0))}-
1\bigr\Vert \leq \tau \sup_{\eta \tau+s_0\leq x\leq s_0} |\nu (x)-\nu(s_0)|,
\end{align}
or 
\begin{align}\label{mars22}
 \bigl \Vert U\bigl(\tau+\frac{s_{0}}{\eta}, \frac{s_{0}}{\eta}\bigr)-
e^{-i \tau h(\nu(s_0))}\bigr \Vert \leq \tau \sup_{\eta \tau+s_0\leq x\leq s_0} |\nu (x)-\nu(s_0)|,
\end{align}
where the right hand side goes to zero with $\eta$. 
The same type of norm estimate as \eqref{mars22} holds when we replace $U$ with $U_0$ and $h$ with $h_L$. Being a norm estimate, it also holds for 
the adjoints. 

Using \eqref{mars18}, we obtain the norm estimate: 
\begin{align}\label{mars21'}
\bigl \Vert h_TU_0
\bigl (\tau+\frac{s_{0}}{\eta}, \frac{s_{0}}{\eta}\bigr )
P_{{\rm ac}}(h_L)F\bigr \Vert \leq C\langle \tau\rangle ^{-2}. 
\end{align}
Now using the left continuity of $\nu$ at $s_0$, the $L^1$ upper bound in \eqref{mars21'} and the pointwise norm convergence 
of \eqref{mars22}, the Lebesgue dominated convergence theorem gives us
\eqref{mars15}:
\begin{align*}
\lim_{\eta\searrow 0}&
U\bigl(\frac{s_{-1}}{\eta}, \frac{s_{0}}{\eta}\bigr)^*
\Omega_i  e^{-i\frac{s_{-1}-s_0}{\eta}h_L } e^{\frac{i}{\eta}\Delta\nu  \Pi_-}P_{{\rm ac}}(h_L)F\\
&= P_{{\rm ac}}(h_L)F\\
&\quad+i\int_{0}^{-\infty}d\tau  e^{i \tau h(\nu(s_0))}
h_T   e^{-i \tau (h_L+\nu(s_0)\Pi_-)} P_{{\rm ac}}(h_L)F\\
&=\Omega_f P_{{\rm ac}}(h_L)F,
\end{align*}
where the last equality comes from the Dyson equation satisfied by
$\Omega_f$. Then since $U (\frac{s_{-1}}{\eta},
  \frac{s_{0}}{\eta} )$ is unitary, we have:
$$
U\bigl (\frac{s_{-1}}{\eta}, \frac{s_{0}}{\eta}\bigr )\Omega_f P_{{\rm ac}}(h_L)F-
\Omega_i  e^{-i\frac{s_{-1}-s_0}{\eta}h_L }  e^{\frac{i}{\eta}\Delta\nu \Pi_-}P_{{\rm ac}}(h_L)F=o(1)
$$
which is exactly \eqref{mars16}. Finally, if we use in \eqref{mars16} a vector $F=\Omega_f^*P_{{\rm ac, f}}\Psi$ where 
$\Psi$ is some arbitrary unit vector, then $\Omega_f P_{{\rm ac}}(h_L)F=P_{{\rm ac, f}}\Psi$ and \eqref{mars16'} follows.

\end{proof}

\subsection{Proof of \eqref{mars341}}

\begin{proof} Note that it is enough to consider rank one observables, thus we will assume 
without loss of generality that $A=|\Psi\rangle\langle \Psi|$. 
Let us remember the simplifying notation:
\begin{align}\label{mars11}
 P_{{\rm ac}}(0)&:=P_{{\rm ac}}(h(v(0))), 
&P_{{\rm ac}}(-1)&:=P_{{\rm ac}}(h(v(-1)))=P_{{\rm ac}}(h),\nonumber \\
P_{{\rm d}}(0)&:=P_{{\rm d}}(h(v(0))), 
&P_{{\rm d}}(-1)&:=P_{{\rm d}}(h(v(-1)))=P_{{\rm d}}(h).
\end{align}
We can write:
\begin{align}\label{mars30}
 U(-\eta^{-1},0)^*&f_{\rm eq}(h)U(-\eta^{-1},0)A\nonumber\\
 &=
U(-\eta^{-1},0)^*f_{\rm eq}(h)P_{{\rm d}}(-1)U(-\eta^{-1},0)P_{{\rm d}}(0)A\nonumber \\
&\quad+U(-\eta^{-1},0)^*f_{\rm eq}(h)P_{{\rm ac}}(-1)U(-\eta^{-1},0)P_{{\rm ac}}(0)A
 \nonumber \\
&\quad+
U(-\eta^{-1},0)^*f_{\rm eq}(h)P_{{\rm d}}(-1)U(-\eta^{-1},0)P_{{\rm ac}}(0)A\nonumber \\
&\quad+
U(-\eta^{-1},0)^*f_{\rm eq}(h)P_{{\rm ac}}(-1)U(-\eta^{-1},0)P_{{\rm d}}(0)A
.
\end{align}
We will treat these four terms separately.

\subsubsection{The discrete-discrete term}
From now on the notation $T_1\sim T_2$ will mean that the difference $T_1- T_2$
goes to zero with $\eta$ in the appropriate topology. 
Under our assumptions on $v$, the discrete eigenvalues of the instantaneous Hamiltonian 
$h(s)=h(v(s))$ remain at a positive distance from the continuous spectrum
if $-1\leq s< s_{c}$ and $s'_{c}< s\leq 0$. Moreover, we assumed that
$v''(s)$ is continuous and bounded on the open intervals. Then using the usual adiabatic theorem
for eigenstates, the continuity and the group properties of $U$, we obtain that in the operator 
norm topology we have: 
\begin{align}\label{mars31}
 P_{{\rm d}}(-1)U(-1/\eta,s_{c}/\eta)&\sim U(-1/\eta,s_{c}/\eta)P_{\rm d}(s_{c,1}-0), \nonumber
\\ U(s'_{c}/\eta,0)P_{{\rm d}}(0)&\sim P_{\rm d}(s'_{c}+0)U(s'_{c}/\eta,0),\nonumber \\
P_{{\rm d}}(-1)U(-\eta^{-1},0)P_{{\rm d}}(0)
&\sim \nonumber\\
U(-1/\eta,s_{c}/\eta)P_{\rm d}(s_{c}&-0) U(s_{c}/\eta, s'_{c}/\eta)
P_{\rm d}(s'_{c}+0)U(s'_{c}/\eta,0).
\end{align}
We see that in the middle of the last line we get the rank one operator $$P_{\rm
  d}(s_{c}-0) U(s_{c}/\eta, s'_{c}/\eta)
P_{\rm d}(s'_{c}+0).$$
We will now show that 
$$\lim_{\eta\searrow 0} \langle \psi(s_c-0)\,|\, U(s_{c}/\eta, s'_{c}/\eta)\psi(s'_c+0)\rangle=0.$$
In order to do this, let us go back to Proposition \ref{propestim} and identify $s_{-1}=s_{c}$,  $s_0=s'_{c}$ and $\nu$ with $v$ 
restricted to $[s_{c},s'_{c}]$. We see that the 
potential $v$ is continuous from the left at $s'_{c}$, and moreover, the spectrum of the instantaneous Hamiltonian 
$h(s)$ is purely absolutely continuous if $s\in [s_{c},s'_{c}]$. 
Remember that the wave operators $\Omega(s'_c-0)$ and $\Omega(s_c+0)$ defined in 
\eqref{mf2} map ${\Ran}(P_{{\rm ac}}(h_L))$ onto
$\mathcal{H}$. Comparing with \eqref{mars16'}, we can 
identify $\Omega_f$ with $\Omega(s'_c-0)$ and $\Omega_i$ with $\Omega(s_c+0)$, 
while ${\Ran}P_{{\rm ac, f}}=\mathcal{H}$. Thus:
\begin{align*}
\langle \psi(s_c-0) \,|&\,U(s_{c}/\eta, s'_{c}/\eta)\psi(s'_c+0)\rangle\nonumber\\ \sim 
\langle \Omega(s_c+0)^*\psi(s_c-0)\,|&\,
e^{-i\frac{s_{c}-s'_{c}}{\eta}h_L}e^{\frac{i\Delta v}{\eta}\Pi_-}\Omega(s'_c-0)^*\psi(s'_c+0)\rangle.
\end{align*}
The $\eta$ dependence on the right hand side is explicit, and 
this term converges to zero with $\eta$. It means that there is no 
contribution to the final steady state from the purely discrete part.

\subsubsection{The continuous-continuous term}
The next term in \eqref{mars30} contains both the initial and final projections on the instantaneous 
absolutely continuous subspaces. Note the important thing 
that $P_{\rm ac}(-1)=P_{\rm ac}(h)$ commutes with $f_{\rm eq}(h)$. 

We can again apply Proposition \ref{propestim}, where now $s_{-1}=-1$, $s_0=0$, $\Omega_i=\Omega(-1)$,  
$\Omega_f=\Omega(0)$ and $P_{\rm ac, f}=P_{\rm ac}(0)$. Using \eqref{mars16'} we get: 
\begin{align}\label{mars36}
U(-\eta^{-1},0)^*&f_{\rm eq}(h)P_{\rm ac}(-1)U(-\eta^{-1},0)P_{\rm ac}(0)\Psi\nonumber \\
&\sim U(-\eta^{-1},0)^*f_{\rm eq}(h)P_{\rm ac}(-1)\Omega(-1)
e^{\frac{i}{\eta}h_L}\nonumber\\
&\qquad\cdot e^{i\frac{\Delta
    \nu}{\eta}\Pi_-}P_{{\rm ac}}(h_L)\Omega(0)^*P_{\rm ac}(0)\Psi.
\end{align}
The important thing is that $\Omega(-1)$ intertwines between $h$ and $h_L$. Thus  the right hand side of the above 
equation is equal to:
$$
U(-\eta^{-1},0)^*P_{\rm ac}(-1)\Omega(-1)
e^{\frac{i}{\eta}h_L} e^{i\frac{\Delta
    \nu}{\eta}\Pi_-}f_{\rm eq}(h_L)P_{{\rm ac}}(h_L)\Omega(0)^*P_{\rm ac}(0)\Psi.
    $$
But now we can apply \eqref{mars15} and obtain the result:
\begin{align}\label{mars37}
U(-\eta^{-1},0)^*&f_{\rm eq}(h)P_{\rm ac}(-1)U(-\eta^{-1},0)P_{\rm ac}(0)\Psi
\nonumber\\
&\sim \Omega(0) f_{\rm eq}(h_L)P_{{\rm ac}}(h_L)\Omega(0)^*P_{\rm ac}(0)\Psi.
\end{align}
Thus the first nonzero contribution to $\langle A\rangle$ in \eqref{mars10} is:
\begin{align}\label{mars37'}
 \langle \Omega(0)^*P_{\rm ac}(0)
  \Psi\, |\,
f_{\rm eq}(h_L)
P_{{\rm ac}}(h_L)\Omega(0)^*P_{\rm ac}(0)\Psi  \rangle .
\end{align}

\subsubsection{The mixed terms}
The third contribution in \eqref{mars30} is:
$$
U(-\eta^{-1},0)^*f_{\rm eq}(h)P_{\rm d}(-1)U(-\eta^{-1},0)P_{\rm ac}(0)\Psi.
$$
A direct application of \eqref{mars16'} with $s_{-1}=-1$ and $s_0=0$
shows that the vector $U(-\eta^{-1},0)P_{\rm ac}(0)\Psi$ lies 
almost completely in ${\Ran}P_{\rm ac}(-1)$, thus when
projected on ${\Ran}P_{\rm d}(-1)$ it will converge to
zero with $\eta$. Thus this term will also disappear.

The fourth contribution is probably the most interesting one. Using
the decomposition $1=P_{\rm ac}(0)+P_{\rm d}(0)$ to the left, we can write:
\begin{align}\label{mars38}
U(-\eta^{-1}&,0)^*f_{\rm eq}(h)P_{\rm ac}(-1)U(-\eta^{-1},0)P_{\rm d}(0)A \\
&=
P_{\rm ac}(0)U(-\eta^{-1},0)^*f_{\rm eq}(h)P_{\rm ac}(-1)U(-\eta^{-1},0)P_{\rm d}(0)A\nonumber \\
&\quad+P_{\rm d}(0)U(-\eta^{-1},0)^*f_{\rm eq}(h)P_{\rm ac}(-1)U(-\eta^{-1},0)P_{\rm d}(0)A.\nonumber
\end{align}
We show that the trace of the first term will converge to zero. After taking the trace with $A=\vert \Psi\rangle \langle \Psi\vert $ we obtain:
\begin{align}\label{mars39}
\langle \Psi &\,|\, P_{\rm ac}(0)U(-\eta^{-1},0)^*f_{\rm eq}(h)P_{\rm ac}(-1)U(-\eta^{-1},0)P_{\rm d}(0)
\Psi\rangle \nonumber \\
&=
\langle U(-\eta^{-1},0)^*f_{\rm eq}(h)P_{\rm ac}(-1) U(-\eta^{-1},0) P_{\rm ac}(0)\Psi \,|\, P_{\rm d}(0)\Psi\rangle.
\end{align}
But the 'bra' vector in the second line is the same as the one in \eqref{mars37}, which
we know that asymptotically enters the range of $\Omega(0)$, thus
becomes orthogonal to the 'ket' vector.

Now let us treat the last contribution. Because  $P_{\rm d}(0)$ is one dimensional, the trace of this last contribution is:
\begin{align}\label{mars40}
\langle \psi(0)\vert& A\psi(0)\rangle  \nonumber\\
&\cdot
\Tr\{P_{\rm d}(0)U(-\eta^{-1},0)^*f_{\rm eq}(h)P_{\rm ac}(-1)U(-\eta^{-1},0)P_{\rm d}(0)\} .
\end{align}

Using the group property
$U(-\eta^{-1},0)=U(-\eta^{-1},s'_{c}/\eta)U(s'_{c}/\eta,0)$, then applying the 
adiabatic theorem which says that in operator norm:
\begin{align*}
U(s'_{c}/\eta,0)P_{\rm d}(0)&\sim P_{\rm
  d}(s'_{c}+0)U(s'_{c}/\eta,0),
  \\
   P_{\rm d}(0)U(s'_{c}/\eta,0)^*&\sim U(s'_{c}/\eta,0)^*
  P_{\rm d}(s'_{c}+0),
  \end{align*} 
  and finally the trace cyclicity, we obtain:
\begin{align}\label{mars42}
&\Tr\{P_{\rm d}(0)U(-\eta^{-1},0)^*f_{\rm eq}(h)P_{\rm ac}(-1)U(-\eta^{-1},0)P_{\rm d}(0)\} 
\nonumber \\
&\sim \langle \psi(s'_c+0)| U(-\eta^{-1},s'_{c}/\eta)^* f_{\rm eq}(h)P_{\rm ac}(-1)U(-\eta^{-1},s'_{c}/\eta)\psi(s'_c+0)\rangle.
\end{align}
Now we can reason as in the continuous-continuous case, where
$s_{-1}=-1$, $s_0=s_c'$, and $\Omega(s'_c-0)$ plays the
role of $\Omega_f$, which ends the proof of \eqref{mars341}. 

Regarding \eqref{mars3411}, the only
difference is that the discrete-discrete term will now contribute
because we can apply the usual adiabatic theorem on the whole
evolution interval, while the mixed terms disappear. The
continuous-continuous term is identical with the previous one. 
\end{proof}

\section{Proof of Proposition \ref{propodilo}}\label{section4}
First of all, the result stated in { (i)} is a trivial consequence of the theorem. We will only prove { (ii)}  in detail 
since the proof of { (iii)} is very much similar. 

Let us now focus on the case described in Figure \ref{f1}. The family $h(v)$ is norm continuous in $v$ and its spectrum is
Lipschitz continuous in $v$. The spectrum of $h(v_{c,1})$ is purely absolutely continuous
and consists of two well isolated bands: $[-2,2]\cup
[-2+v_{c,1},2+v_{c,1}]$. We can find a positively oriented circle $\Gamma$
which completely includes the interval $[-3+v_{c,1},3+v_{c,1}]$ while $[-2,2]$ lies
outside. Moreover, if $\delta_0>0$ is small enough and $|v-v_{c,1}|\leq \delta_0$ then the spectrum of $h(v)$ will be at a distance proportional with $\delta_0$ from the spectrum of $h(v_{c,1})$ and:
\begin{align}\label{mars42a}
\sup_{z\in \Gamma}\;\sup_{|v-v_{c,1}|\leq \delta_0}\|(h(v)-z)^{-1}\|<\infty.
\end{align}
The above bound also holds true if $h(v)$ is replaced with
$h_L(v)=h_L+v\Pi_-$. 

Using the second resolvent identity, we have that 
$$
\Pi_-\{(h(v)-z)^{-1}-(h_L(v)-z)^{-1}\}=-\Pi_-(h_L(v)-z)^{-1}h_T(h(v)-z)^{-1}, 
$$
for $z\in \Gamma$.
We now argue that if 
 $\Pi_-^{(M)}=\sum_{m\geq M}|m_-\rangle\langle m_-|$ is $\Pi_-$ without
its first $M$ sites, then:
\begin{equation}\label{mars43}
\lim_{M\to\infty}\{\sup_{z\in \Gamma}\;\sup_{|v-v_{c,1}|\leq \delta_0}\|\Pi_-^{(M)}(h_L(v)-z)^{-1}h_T(h(v)-z)^{-1}\|\}=0.
\end{equation}
This is due to the fact that $h_T$ is localized while, uniformly in $z\in\Gamma$,  the two resolvents 
have exponentially localized kernels near the diagonal.

By integrating $(h_L(v)-z)^{-1}$ over $\Gamma$ we obtain $\Pi_-$. Thus the Riesz 
projection $P(v)$ corresponding to the spectrum of $h(v)$ contained inside
$\Gamma$ obeys the estimate:
\begin{align}\label{mars44}
\lim_{M\to\infty}\{\sup_{|v-v_{c,1}|\leq \delta_0}\|\Pi_-^{(M)}\{P(v)-\Pi_-\}\|\}=0.
\end{align}
Now fix some $\varepsilon>0$. We can find $M=M_\varepsilon$ such that:
\begin{align}\label{mars44'}
\sup_{|v-v_{c,1}|\leq
  \delta_0}\|\Pi_-^{(M_\varepsilon)}P(v)-\Pi_-^{(M_\varepsilon)}\|\leq
\frac{\varepsilon}{6 \|f_{\rm eq}\|_\infty}.
\end{align}
In order to shorten notation, let us write $\psi_{c,\delta}$ instead of $\psi(s'_c+0)$. This is the eigenvector corresponding to the bias $v=v_{c,1}-\delta$, and the eigenvalue $\lambda(v)$ is only barely larger than $v+2$, their difference being proportional to $\delta^2$, see \eqref{mars1}. 

 From \eqref{mars5}, \eqref{mars0}  and the fact that $\psi_{c,\delta}$ converges weakly to zero as $\delta \rightarrow 0$ (see the end of Section \ref{sect22}) we see that
$\psi_{c,\delta}$ becomes more on more delocalized, far away on the left lead, when $\delta$ tends to zero. In other words, for
every fixed $M\geq 0$ we have:
\begin{align}\label{mars41}
\lim_{\delta\searrow 0}\|\psi_{c,\delta}-\Pi_-^{(M)}\psi_{c,\delta}\|=0.
\end{align}
Thus there exists $\delta_\varepsilon$ small enough such that for every
$\delta<\delta_\varepsilon$ we have:
\begin{align}\label{mars411}
\|\psi_{c,\delta}-\Pi_-^{(M_\varepsilon)}\psi_{c,\delta}\|\leq
\frac{\varepsilon}{6\|f_{\rm eq}\|_\infty}.
\end{align}
Again in order to shorten notation, let us write $\Omega_{c,\delta}$ instead of $\Omega(s'_c-0)$. 
Going back to \eqref{mars40a} and using \eqref{mars411} we obtain:
\begin{align*}
|\langle \psi_{c,\delta}\,|\,&
  \Omega_{c,\delta} f_{\rm eq}(h_L)P_{{\rm
      ac}}(h_L)\Omega_{c,\delta}^*\psi_{c,\delta}\rangle
      \\
      &-  \langle \psi_{c,\delta}\,|\,
  \Pi_-^{(M_\varepsilon)}\Omega_{c,\delta} f_{\rm eq}(h_L)P_{{\rm
      ac}}(h_L)\Omega_{c,\delta}^*\psi_{c,\delta}\rangle |\leq \frac{\varepsilon}{6}
\end{align*}
which together with \eqref{mars44'} and \eqref{mars411} gives:
\begin{align}\label{mars45}
 | \langle \psi_{c,\delta}\,|\,&
  \Omega_{c,\delta}f_{\rm eq}(h_L)P_{{\rm
      ac}}(h_L)\Omega_{c,\delta}^*\psi_{c,\delta}\rangle
      \nonumber\\
      &-  \langle \psi_{c,\delta}\,|\,
  P(v)\Omega_{c,\delta} f_{\rm eq}(h_L)P_{{\rm
      ac}}(h_L)\Omega_{c,\delta}^*\psi_{c,\delta}\rangle  |\leq \frac{\varepsilon}{2},
\end{align}
for every $|v-v_{c,1}|\leq \delta_0$ and $\delta<\delta_\varepsilon$. 

The idea is to take $v=v_{c,1}+\delta$ in the above estimate, corresponding to $v(s'_c-0)$ for which the Hamiltonian has purely absolutely continuous spectrum. Before that, let us notice a few identities. Using the
intertwining property of $\Omega_{c,\delta}$ we have that
$P(v_{c,1}+\delta)\Omega_{c,\delta}=\Omega_{c,\delta}\Pi_-$
where $\Pi_-$ is nothing but the spectral projection of the
decoupled operator $h_0+(v_{c,1}+\delta)\Pi_-$ corresponding to the absolutely continuous spectrum contained in
$\Gamma$. Defining $f_{\rm eq,\delta}(x):=f_{\rm
  eq}(x-v_{c,1}-\delta)$ we have the identity:
$$\Pi_-f_{\rm eq}(h_L)=\Pi_-f_{\rm eq,\delta}(h_L+(v_{c,1}+\delta)\Pi_-),$$
or putting everything together:
\begin{align}\label{mars46}
P(v_{c,1}+\delta)&\Omega_{c,\delta}f_{\rm eq}(h_L)P_{{\rm
      ac}}(h_L)\Omega_{c,\delta}^*\nonumber\\
      &=P(v_{c,1}+\delta)\Omega_{c,\delta}f_{\rm
    eq,\delta}(h_L+(v_{c,1}+\delta)\Pi_-)\Omega_{c,\delta}^*\nonumber \\
&=P(v_{c,1}+\delta)f_{\rm
    eq,\delta}(h(v_{c,1}+\delta)),
\end{align}
where the second identity is again implied by the intertwining
properties of $\Omega_{c,\delta}$; remember that $h(v_{c,1}+\delta)$
has purely absolutely continuous spectrum. 

Thus putting $v=v_{c,1}+\delta$ in \eqref{mars45} and using the
identity \eqref{mars46}, we obtain that for every
$\delta<\delta_\varepsilon$ we have:
\begin{align}\label{mars47}
| \langle \psi_{c,\delta}\,|\,&
  \Omega_{c,\delta}f_{\rm eq}(h_L)P_{{\rm
      ac}}(h_L)\Omega_{c,\delta}^*\psi_{c,\delta}\rangle
      \nonumber \\
      &-  \langle \psi_{c,\delta}\,|\,P(v_{c,1}+\delta)f_{\rm
    eq,\delta}(h(v_{c,1}+\delta))\psi_{c,\delta}\rangle  |\leq \frac{\varepsilon}{2}.
\end{align}
Using again \eqref{mars411} and \eqref{mars44'} we can replace 
$P(v_{c,1}+\delta)$  in \eqref{mars47} with the identity operator, with a cost of an error of at
most $\varepsilon/3$. Thus for every $\delta<\delta_\varepsilon$ we have:
\begin{align}\label{mars48}
 | \langle \psi_{c,\delta}\,|\,&
  \Omega_{c,\delta}f_{\rm eq}(h_L)P_{{\rm
      ac}}(h_L)\Omega_{c,\delta}^*\psi_{c,\delta}\rangle
      \nonumber\\
      &-  \langle \psi_{c,\delta}\,|\,f_{\rm
    eq,\delta}(h(v_{c,1}+\delta))\psi_{c,\delta}\rangle  |\leq \frac{5\varepsilon}{6}.
\end{align}
Since the family $h(v)$ is norm continuous and $f_{\rm
    eq}$ is supposed to be continuous, we may find a
$\delta'_\varepsilon<\delta_\varepsilon$ such that 
$$\|f_{\rm
    eq,\delta}(h(v_{c,1}+\delta))-f_{\rm
    eq,\delta}(h(v_{c,1}-\delta))\|\leq \varepsilon/100,\quad \delta<\delta'_\varepsilon.$$
But $\psi_{c,\delta}$ is the eigenvector of $h(v_{c,1}-\delta)$
corresponding to the eigenvalue $\lambda(v_{c,1}-\delta)$. Thus:
\begin{align*} | \langle \psi_{c,\delta}\,|\,&
  \Omega_{c,\delta}f_{\rm eq}(h_L)P_{{\rm
      ac}}(h_L)\Omega_{c,\delta}^*\psi_{c,\delta}\rangle\\
      &-
f_{\rm
    eq}(\lambda(v_{c,1}-\delta)-v_{c,1}-\delta)  |\leq
\frac{6\varepsilon}{7},\quad \forall \delta<\delta'_\varepsilon.
\end{align*}
But $\lambda(v_{c,1}-\delta)-v_{c,1}-\delta$ converges to $2$ when
$\delta$ goes to zero, thus there exists $\delta^{''}_\varepsilon<\delta'_\varepsilon$ such that for every 
$\delta<\delta^{''}_\varepsilon$ we have:
\begin{equation*}
| \langle \psi_{c,\delta}\,|\,
  \Omega_{c,\delta}f_{\rm eq}(h_L)P_{{\rm
      ac}}(h_L)\Omega_{c,\delta}^*\psi_{c,\delta}\rangle-
f_{\rm
    eq}(2)  |<
\varepsilon.
\end{equation*}
The proof is over.
\qed

\subsection*{Acknowledgments} The authors acknowledge support from the
Danish FNU grant {\it Mathematical Analysis of Many-Body Quantum Systems}. Part of this work was done at Institut Mittag-Leffler during 
the program {\it Hamiltonians in Magnetic Fields}.

\end{document}